\documentclass[aap]{imsart}

\RequirePackage[OT1]{fontenc}
\RequirePackage{amsthm,amsmath}
\RequirePackage[authoryear, round]{natbib}
\RequirePackage{hyperref}

\usepackage{amssymb}
\usepackage{mathrsfs}
\usepackage{graphicx}
\usepackage{url}
\usepackage{array}
\usepackage{enumerate}
\usepackage[utf8]{inputenc}
\usepackage{times}
\usepackage[mathscr]{eucal}
\usepackage[math]{easyeqn}
\usepackage{xifthen}
\usepackage{nomencl}
\makenomenclature

\startlocaldefs
\numberwithin{equation}{section}
\theoremstyle{plain}
\newtheorem*{theorem*}{Theorem}
\newtheorem{theorem}{Theorem}
\newtheorem{corollary}{Corollary}
\newtheorem{proposition}{Proposition}
\newtheorem{lemma}{Lemma}
\newtheorem{Prop}{Proposition}

\newtheorem{Rem}{Remark}
\numberwithin{equation}{section}
\theoremstyle{plain}

\theoremstyle{plain}

\endlocaldefs

\renewcommand{\(}{$\,}
\renewcommand{\)}{\,$}

\renewcommand{\bar}[1]{\overline{#1}}
\renewcommand{\hat}[1]{\widehat{#1}}
\renewcommand{\tilde}[1]{\widetilde{#1}}

\renewcommand{\Gamma}{\varGamma}
\renewcommand{\Pi}{\varPi}
\renewcommand{\Sigma}{\varSigma}
\renewcommand{\Delta}{\varDelta}
\renewcommand{\Lambda}{\varLambda}
\renewcommand{\Psi}{\varPsi}
\renewcommand{\Phi}{\varPhi}
\renewcommand{\Theta}{\varTheta}
\renewcommand{\Omega}{\varOmega}
\renewcommand{\Xi}{\varXi}
\renewcommand{\Upsilon}{\varUpsilon}

\def\P{I\!\!P}

\def\kappa{\varkappa}

\def\mub{\bar{\mu}}
\def\mun{\bar{\mu}_n}
\def\Sb{\bar{S}}

\usepackage{color}

\definecolor{blue(pigment)}{rgb}{0.2, 0.2, 0.6}
\definecolor{ultramarine}{rgb}{0.07, 0.04, 0.56}
\definecolor{darkspringgreen}{rgb}{0.09, 0.45, 0.27}
\definecolor{hookersgreen}{rgb}{0.0, 0.44, 0.0}
\definecolor{plum(traditional)}{rgb}{0.56, 0.27, 0.52}
\definecolor{purple(html/css)}{rgb}{0.5, 0.0, 0.5}
\definecolor{magenta(dye)}{rgb}{0.79, 0.08, 0.48}

\renewcommand{\(}{$\,}
\renewcommand{\)}{\,$}

\def\P{I\!\!P}

\def\kappa{\varkappa}

\def\mub{\bar{\mu}}
\def\mun{\bar{\mu}_n}
\def\Sb{\bar{S}}

\begin{document}
\begin{frontmatter}
\title{Gaussian processes with multidimensional distribution inputs via optimal transport and Hilbertian embedding}
\runtitle{Gaussian processes via optimal transport}

\begin{aug}
\author{\fnms{Fran\c{c}ois} \snm{Bachoc}\ead[label=e1]{francois.bachoc@math.univ-toulouse.fr}}
\address{Institut de Math\'ematiques de Toulouse\\
	\printead{e1}}

\author{\fnms{Alexandra} \snm{Suvorikova}\ead[label=e2]{suvorikova@math.uni-potsdam.de}}

\address{Potsdam University\\
	\printead{e2}}

\author{\fnms{David} \snm{Ginsbourger}
	\ead[label=e3]{david.ginsbourger@stat.unibe.ch}}

\address{Idiap Research Institute\\
	University of Bern\\
	\printead{e3}}

\author{\fnms{Jean-Michel} \snm{Loubes}
	\ead[label=e4]{loubes@math.univ-toulouse.fr}}

\address{Institut de Math\'ematiques de Toulouse\\
	\printead{e4}}

\author{\fnms{Vladimir} \snm{Spokoiny}
	\ead[label=e5]{vladimir.spokoiny@wias-berlin.de}}

\address{Weierstrass Institute for Applied Analysis and Stochastics\\
	\printead{e5}}

\runauthor{Bachoc, Suvorikova, Ginsbourger, Loubes, Spokoiny}


\end{aug}

\begin{abstract}
In this work, we investigate Gaussian Processes indexed by multidimensional distributions. While directly constructing radial positive definite kernels based on the Wasserstein distance has been proven to be possible in the unidimensional case, such constructions do not extend well to the multidimensional case as we illustrate here. To tackle the problem of defining positive definite kernels between multivariate distributions based on optimal transport, we appeal instead to Hilbert space embeddings relying on optimal transport maps to a reference distribution, that we suggest to take as a Wasserstein barycenter. We characterize in turn radial positive definite kernels on Hilbert spaces, and show that the covariance parameters of virtually all parametric families of covariance functions are microergodic in the case of (infinite- dimensional) Hilbert spaces. We also investigate statistical properties of our suggested positive definite kernels on multidimensional distributions, with a focus on consistency when a population Wasserstein barycenter is replaced by an empirical barycenter and additional explicit results in the special case of Gaussian distributions. Finally, we study the Gaussian process methodology based on our suggested positive definite kernels in regression problems with multidimensional distribution inputs, on simulation data stemming both from synthetic examples and from a mechanical engineering test case.
\end{abstract}





\end{frontmatter}

\section{Introduction}
Gaussian process models are widely used in fields such as geostatistics, computer experiments and machine learning \cite{rasmussen06gaussian}, \citet{santner03design}. In a nutshell, Gaussian process modelling consists in assuming for an unknown function of interest to be one realisation of a Gaussian process, or equivalently of a Gaussian random field indexed by the source space of the objective function, and is often cast as part of the Bayesian arsenal for non-parametric estimation in function spaces. For instance, in computer experiments, the input points of the function are simulation input parameters and the output values are quantities of interest obtained from simulation responses.  Furthermore, there has been a huge amount of literature dealing with the use of Gaussian Processes in Machine Learning over the last decade. We refer for instance to \citet{rasmussen2004gaussian}, \citet{scholkopf2002learning} or \citet{cristianini2000support} and references therein.

Gaussian process models heavily rely on the specification of a covariance function, or ``kernel'', that characterises linear dependencies between values of the process at different observation points. In fact, the kernel, which can be seen as a similarity measure between locations in index space, also induces a (pseudo-)metric on the index space often referred to as the ``canonical metric associated with the kernel'' via the variogram function of geostatisticians. A natural question for a given kernel is how those inherently associated notions of similarity/dissimilarity interplay with prescribed metrics on the index space. In Euclidean space, one often speaks of radial kernel for those covariance functions that are explicitly depending on the Euclidean distance between points. Radial kernels with respect to other metrics have also be investigated, see e.g. kernels writing as functions of the $\ell^1$ distance in multivariate Euclidean spaces \citet{wendland2004scattered}. 

In this paper we consider Gaussian processes indexed by distributions supported on $\mathbb{R}^p$, and we investigate ways to build positive definite kernels based on the Wasserstein distance. Distributional inputs can occur in a number of practical situations and exploring admissible kernels for using Gaussian Process and related methods in this context is a pressing issue. Situations of that kind include the case of uncertain vector inputs to a vector-to-scalar deterministic function, but also a variety of other settings such as histogram inputs standing for instance for ratings from a panel of experts,  compositional data in geosciences, or randomized strategies in a Bayesian game-theoretic framework.  

In some situations, distribution-valued inputs may arise as a convenient way to describe complex objects and media, e.g. a number of physical simulations require maps or parameter fields as inputs, and in some cases it can be beneficial to reparametrize them so as to work with probabiliy distributions. For instance in \citet{ginsbourger2016design}, the computer model \citet{CASTEM} is studied, where the input simulation parameter consists of a set of disks located on a unit square $[0,1]^2$, modelling a material, for which a stress measure is associated. A Gaussian process model on distributions enables to treat the input sets of disks as measures, and to model the stress values as stemming from a random field indexed by the input distributions.

In this framework, a natural aim is to construct covariance functions for Gaussian processes indexed by such inputs, that is constructing positive definite kernels on sets of probability measures.

The simplest method is perhaps to compare a set of parametric features built from the probability distributions, such as the mean or the higher moments. This approach is limited, as the effect of such parameters does not take the whole distribution into account. Specific positive definite kernels should be designed in order to map distributions into a reproducing kernel Hilbert space in which the whole arsenal of kernel methods can be extended to probability measures. This issue has recently been considered in \citet{2016arXiv160509522M} or \citet{KolouriZouRohde}. We aim at basing these kernels on the Wasserstein, or transport-based, distance which was shown to be relevant and insightful for comparing or studying distributions \citet{villani2009optimal,chernozhukov2017monge,peyre2019computational}.

This issue has been studied for the one dimensional case in \citet{myieee} or in \citet{2018arXiv180610493T}, using the special expression of the Wasserstein distance in dimension 1. Yet this case uses  the property of the optimal coupling with the uniform random variable which is very specific to the one dimensional case. The positive definite kernels provided in the one dimensional case are not positive definite any longer, when they are extended to higher dimensions, as we illustrate numerically in Section \ref{s:num}.\\

\indent In the general dimension case, in order to build a positive definite kernel from the Wasserstein distance, we associate to each input distribution its optimal transport map to a reference distribution. We then provide positive definite kernels on the Hilbert space corresponding to these optimal transport maps. This results in a positive definite kernel for multidimensional distributions. As a reference distribution, we recommend to take the empirical Fr\'echet mean (or barycenter) of the distributions. We remark that the notion of Wasserstein barycenters and their use in machine learning and in statistics has been tackled  recently in, for instance, \citet{agueh2011barycenters}, \citet{bigot2012characterization}, \citet{boissard2015distribution}.  Although computational aspects of optimal transports are a difficult issue, substantial work has been conducted to provide feasible algorithms to compute barycenters and optimal transport maps, see for instance \citet{kroshnin2019statistical}, ~\citet{uribe2018distributed}, or  \citet{peyre2019computational} and references therein.
Thus our suggested procedure is feasible in practice, as is confirmed by our simulation results on simulated data and on the data from the CASTEM computer model \citet{CASTEM,ginsbourger2016design}.

We also characterize all the continuous radial positive definite kernels on Hilbert spaces. This is carried out by showing that they coincide with continuous radial positive definite kernels on Euclidean spaces of arbitrary dimension, and by revisiting existing results for the Euclidean case \citet{wendland2004scattered}. In addition, we show that when considering parametric families of covariance functions for Gaussian processes on infinite dimensional Hilbert spaces, all the covariance parameters are microergodic in general. Microergodicity is an important concept for the asymptotic analysis of Gaussian processes \citet{stein99interpolation,zhang04inconsistent,And2010}.

\vskip .1in
We provide furthermore statistical results related to our positive definite kernel construction. We study the asymptotic closeness of the two kernels obtained by taking the empirical barycenter and the population barycenter as reference distributions. We obtain additional more quantitative results in the special case of Gaussian input distributions. We also discuss stationarity and universality.

In the aforementioned simulations, we compare the Gaussian process regression model obtained from our suggested positive definite kernels with the distribution regression procedure of \citet{pmlr-v31-poczos13a}. The results show the benefit of our method.

\vskip .1in
The paper falls into the following parts. In Section \ref{s:frame} we recall some definitions on kernels and on the notion of optimal transport, Wasserstein distance and Wasserstein barycenter of distributions . We also provide our positive definite kernel construction.
The analysis of positive definite kernels and Gaussian processes on Hilbert spaces is provided in Section \ref{s:properties}. Section \ref{s:comput} is devoted to the statistical results related to our kernel construction. The simulation results are provided in Section \ref{s:num}. 
Conclusions are discussed in Section \ref{s:4uSasha}.
The  proofs are postponed to the appendix.

\section{Construction of positive definite kernels for distributions with Hilbert space embedding and optimal transport} \label{s:frame}

\subsection{Background} \label{s:smallintro}

Gaussian process models are now widely used in fields such as geostatistics, computer experiments or machine learning \citet{rasmussen06gaussian}, \citet{santner03design}. A Gaussian process model consists in modelling an unknown function as a realisation of a Gaussian process, and hence corresponds to a functional Bayesian framework. For instance, in computer experiments, the input points of the function are simulation parameters and the output values are quantities of interest obtained from the simulations. 
In this paper we focus on Gaussian processes for which the input parameters are in $\mathcal{P}(\mathbb{R}^p)$ the set of distributions supported on $\mathbb{R}^p$. To study such models, Gaussian Processes must be defined over the set of distributions.  \vskip .1in

Let us recall that a Gaussian process $(Y_x)_{x \in E}$ indexed by a set $E$ is entirely characterised by its mean and covariance functions. A covariance function is defined by $(x,y) \in E \times E \to \mathrm{Cov}(Y_x,Y_y)$.
In general, a function $K : E \times E \mapsto \mathbb{R}$ is actually the covariance of a random process if and only if it is a \emph{positive definite kernel}, that is for every $x_1,\cdots,x_n \in E$ and $\lambda_1,\cdots,\lambda_n \in \mathbb{R}, $
\begin{equation} \label{eq:ineq_positive_definite} \sum_{i,j=1}^n \lambda_i \lambda_j K(x_i,x_j) \geq 0. \end{equation}
In this case we say that $K$ is a \emph{covariance kernel}.
If the quadratic form \eqref{eq:ineq_positive_definite} is always strictly positive when $x_1,\ldots,x_n$ are two-by-two distinct, then we say that $K$ is a  \emph{strictly positive definite kernel}.

We also say that $K$ is a \emph{conditionally negative definite kernel} if the quadratic form in \eqref{eq:ineq_positive_definite} is non-positive when $\sum_{i=1}^n \lambda_i = 0$.

The notions of Wasserstein distance and optimal transport will be central to our construction of positive definite kernels on $\mathcal{P}(\mathbb{R}^p)$. Let us introduce them now (see also \citet{villani2009optimal}). 
Let us consider the set $\mathcal{W}_2({\mathbb{R}^p})$ of probability measures on $\mathbb{R}^p$ with finite moments of order two.
For two $\mu,\nu$ in $\mathcal{W}_2\left( \mathbb{R}^p \right)$ , we denote by $\Pi(\mu, \nu)$ the set of all
probability measures $\pi$ over the product set $\mathbb{R}^p \times\mathbb{R}^p $
with first (resp. second) marginal $\mu$ (resp. $\nu$).

The transportation cost with quadratic cost function, or quadratic transportation cost, between these two measures $\mu$ and $\nu$ is defined as

\begin{equation} \label{eq:quadratic_cost}
\mathcal{T}_2(\mu, \nu) = \inf_{ \pi \in \Pi(\mu, \nu)} \int \left\| x - y\right\| ^2 d \pi(x,y).
\end{equation}
In the above display and throughout this paper, we let $\| \cdot \|$ be the Euclidean norm on any Euclidean space.
This transportation cost allows to endow the set $\mathcal{W}_2\left(\mathbb{R}^p\right)$
with a metric by defining  the quadratic  Monge-Kantorovich, or quadratic Wasserstein distance between $\mu$ and $\nu$ as
\begin{equation} \label{eq:wasserstein_distance}
W_2(\mu, \nu)  = \mathcal{T}_2(\mu, \nu) ^{1/2}.
\end{equation}
A probability measure $\pi$ in $\Pi(\mu,\nu)$ realizing the infimum in \eqref{eq:quadratic_cost} is called an optimal coupling.
A random vector $(X_1,X_2)$ with distribution $\pi$ in $\Pi(\mu,\nu)$ realizing this infimum is also called an optimal coupling.

Our aim is to base our suggested covariance functions on the notion of optimal transport.
Indeed, the Wasserstein distance has been shown to be a very useful tool in statistics and machine learning \citet{peyre2019computational,chernozhukov2017monge}.

In the one dimensional case, it is actually possible to create covariance functions which values at $\mu , \nu \in \mathcal{W}_2( \mathbb{R}^p )$ are functions of $W_2(\mu , \nu)$ \citet{myieee}. Indeed, in this case,
using a covariance based on the Wasserstein distance amounts to using the following well-known optimal coupling (see \citet{villani2009optimal}). For all $\mu \in \mathcal{P}(\mathbb{R})$ with finite second order moments, let
\begin{equation} Z_\mu:=F^{-1}_{\mu}(U),\end{equation}
where $F^{-1}_\mu$ is defined as 
$$
F^{-1}_\mu(t) = \inf \{u, F_\mu(u) \geq t \},
$$ 
and denotes the quantile function of the distribution $\mu$ and where $U$ is a uniform random variable on $[0,1]$.
This coupling, given by $\{ Z_{\mu} \}$, can be seen as a non-Gaussian random field indexed by the set of distributions on the real line with finite second order moments.
As such, its variogram 
\begin{equation} (\mu,\nu) \mapsto \mathbb{E}(Z_\mu-Z_\nu)^2 \end{equation}
defines a conditionally negative definite kernel, equal to $W_2^2(\mu,\nu)$ since the coupling $(Z_\mu)$ is optimal. This kernel can be used to construct families of covariance functions based on the one-dimensional Wasserstein distance, see \citet{myieee}. 

In general dimension, however, there is no indication that functions of the form $(\mu , \nu) \to F \left( W_2(\mu,\nu) \right)$, where $F(|\cdot|)$ is a standard covariance function on $\mathbb{R}$, are positive definite kernels. 
For instance, in Section~\ref{s:num} we provide simulations where the function $(\mu , \nu) \to \exp( - W_2(\mu,\nu)^2)$ fails to be a positive definite kernel in the case $p=2$ (while it is indeed a valid kernel when $p=1$, see \citet{myieee}).

To tackle this issue, we will use the notion of Wasserstein barycenters, that we now introduce.
When dealing with a collection of distributions $\mu_1,\dots,\mu_n$,  we can define a notion of variation of these distributions. For any $\nu \in \mathcal{W}_2({\mathbb{R}^p})$, set 
$$ {\rm Var}(\nu)=\sum_{i=1}^n W_2^2(\nu,\mu_i).$$
Finding the distribution minimizing the variance of the distributions has been tackled when defining the notion of barycenter of distributions with respect to Wasserstein's distance in the seminal work of~\citet{agueh2011barycenters}.  More precisely, given $p\geq 1$, they provide conditions to ensure existence and uniqueness of the barycenter of the probability measures $(\mu_i)_{1\leq i \leq n}$  with weights $(\lambda_i)_{1\leq i \leq n}$, i.e. a minimizer of the following criterion
\begin{equation} \label{carlier}
\nu \mapsto \sum_{i=1}^n \lambda_i W_2^2(\nu,\mu_i).
\end{equation}
In the last years several works have studied the empirical properties of the barycenters and their applications to several fields. We refer for instance to \citet{bigot2012characterization,boissard2015distribution} and references therein. Hence the Wasserstein barycenter or Fr\'echet mean of distribution appears to be a meaningful feature to represent the mean variations of a set of distributions.

This notion of Wasserstein barycenter has been recently extended to distributions defined on $\mathcal{W}_2(\mathcal{W}_2(\mathbb{R}^p))$, that is the set of measures on $\mathcal{W}_2(\mathbb{R}^p)$ with finite expected variances.
Let $\mathbb{P}$ be a distribution  in $ \mathcal{W}_2(\mathcal{W}_2(\mathbb{R}^p))$ and consider $\mu_1,\dots,\mu_n$ i.i.d probabilities drawn according to the distribution $\mathbb{P}$. In this framework, the Wasserstein distance between distributions on $\mathcal{W}_2(\mathbb{R}^p)$ is defined, for any $\nu \in \mathcal{W}_2(\mathbb{R}^p)$, as 
\begin{equation}
W_2^2(\mathbb{P},\delta_{\nu})  = \int W^2_2(\nu,\mu) d\mathbb{P}(\mu). \label{thedist}
\end{equation}
If $\tilde{\mu}$ is a random distribution obeying law $\mathbb{P}$, this corresponds to 
$$
W^2_2(\mathbb{P},\delta_{\nu}) = \mathbb{E}_{\{\tilde{\mu} \sim  \mathbb{P}\}} W^2_2(\tilde{\mu},\nu). 
$$
Note that we use the same notations for the Wasserstein distances over distributions in $\mathcal{W}_2(\mathbb{R}^p)$ and over distributions on distributions in $ \mathcal{W}_2(\mathcal{W}_2(\mathbb{R}^p))$. The space $ \mathcal{W}_2(\mathcal{W}_2(\mathbb{R}^p))$ inherits the properties of the space $ \mathcal{W}_2(\mathbb{R}^p)$ and is a good choice for considering asymptotic properties of
Wasserstein barycenteric sequences.

We define (if it exists) the Wasserstein barycenter of $\mathbb{P}$ as a probability measure $\bar{\mu}$ in $\mathcal{W}_2(\mathbb{R}^p)$ such that 
$$
\int W^2_2(\bar{\mu},\mu) d\mathbb{P}(\mu) = \inf \left\{ \int W^2_2(\nu,\mu) d\mathbb{P}(\mu), \: \nu \in \mathcal{W}_2(\mathbb{R}^p)\right\}.
$$ 

First, we point out that the notion of barycenter developed in \eqref{carlier} also
corresponds  to the barycenter of the atomic probability $\mathbb{P}$ on the Wasserstein space, defined by
\[
\mathbb{P}=\sum_{i=1}^n\lambda_i \delta_{\mu_i}.
\]
We also recall some facts on the Wasserstein barycenter that are used in the rest of the paper. The following theorem from \citet{alvarez2015wide} guarantees the existence and uniqueness of this barycenter under some assumptions.

\begin{theorem}[Existence of a Wasserstein Barycenter] \label{th:bar}
Let $\mathbb{P} \in \mathcal{W}_2(\mathcal{W}_2(\mathbb{R}^p))$. Assume that every distribution in the support of $\mathbb{P}$ is absolutely continuous with respect to Lebesgue measure on $\mathbb{R}^p$. Then there exists a unique distribution  $\bar{\mu} \in \mathcal{P}$ defined as 
\begin{equation} \label{bar}
\bar{\mu} = {\rm arg}\min_{ \nu \in \mathcal{W}_2(\mathbb{R}^p) }  \left\{ \int W^2_2(\nu,\mu) d\mathbb{P}(\mu) \right\}.
\end{equation} 
\end{theorem}

Using the expression \eqref{thedist},  we can see that Theorem~\ref{th:bar} can be reformulated as stating the existence of the metric projection of $\mathbb{P}$ onto the subset of $\mathcal{W}_2(\mathcal{W}_2(\mathbb{R}^p))$ composed of Dirac measures.  

Consider a sample of i.i.d random distributions $\mu_i,\: i=1,\dots,n$, drawn from the distribution $\mathbb{P}$ and set ${\bar{\mu}}$ to be its barycenter. Let for fixed $n$, $\bar{\mu}_n$ be the empirical barycenter of the $\mu_1,\dots,\mu_n$, defined as 
$$ 
\sum_{i=1}^n \lambda_i W_2^2(\bar{\mu}_n,\mu_i) = \inf \left\{  \sum_{i=1}^n \lambda_i W_2^2(\nu,\mu_i), \: \nu \in \mathcal{W}_2(\mathbb{R}^p) \right\}, 
$$
with $\lambda_1=...=\lambda_n=1$. 
This empirical barycenter exists and is unique as soon as one of the $\mu_i$ is absolutely continuous w.r.t Lebesgue measure in $\mathbb{R}^p$. \\

The following theorem, from \citet{le2017existence}, states that under uniqueness assumption the empirical Wasserstein barycenter $\bar{\mu}_n$ converges to the population Wasserstein barycenter $\bar{\mu}$. 

\begin{theorem} \label{th:empiricconsisten}
Assume that $\mathbb{P}$ belongs to $\mathcal{W}_2(\mathcal{W}_2(\mathbb{R}^p))$ and that its barycenter is unique. Let $\mu_1,...,\mu_n$ be independently drawn from $\mathbb{P}$ and let $\bar{\mu}_n$ be defined as above.
Then the empirical barycenter $\bar{\mu}_n$ is consistent in the sense that when $n$ goes to infinity we have $$ W_2(\bar{\mu},\bar{\mu}_n) \longrightarrow 0, \: (a.s).$$
\end{theorem}

The above consistency theorem for the empirical barycenter will be useful in Section \ref{s:comput}, where we will compare asymptotically two versions of our positive definite kernel construction: one based on the empirical barycenter and one based on the population barycenter. We now turn to this positive definite kernel construction.

\subsection{Construction of positive definite kernels by Hilbert space embedding of optimal transport maps}
\label{s:kernel}

The positive definite kernel that we suggest here is based on the notion of optimal transport map, that we now introduce.
Consider a reference distribution $\bar{\mu} \in \mathcal{W}_2(\mathbb{R}^p)$, which choice will be discussed below.
For $\mu \in \mathcal{W}_2( \mathbb{R}^p )$, let $T_{\mu} : \mathbb{R}^p \to \mathbb{R}^p$ be the optimal transportation maps defined by 
\[
{T_{\mu}}_\sharp {\mu} = \bar{\mu}
\]
where $f_{\sharp }\pi=\pi \circ f^{-1}$ is the push-forward measure of a function $f$ from a measure $\pi$, and
\[
|| \mathrm{id}  - T_\mu ||_{ L^2( \mu ) } = W_2( \mu , \bar{\mu} ).
\]

Note that the map $T_{\mu}$ is uniquely defined when $\mu$ is absolutely continuous w.r.t. Lebesgue measure. Furthermore, $T_{\mu}$ is invertible from the support of $\mu$ to the support of $\bar{\mu}$ if also $\bar{\mu}$ is absolutely continuous.

\begin{Rem}
We point out that the existence of transportation maps that can be considered as gradients of convex functions  is commonly referred to as Brenier's theorem and originated from Y. Brenier's work in the analysis and mechanics literature in~\citet{brenier1991polar}.
Much of the current interest in transportation problems emanates from this area of mathematics. We conform to the common use of the name. However, it is worthwile pointing out that a similar statement was established earlier independently in a
probabilistic framework in~\citet{cuesta1989notes} : they show existence of an optimal transport map for quadratic cost over Euclidean and Hilbert spaces, and prove monotonicity of the optimal map in some sense (Zarantarello monotonicity).
\end{Rem}

We are now in position to construct a positive definite kernel, by associating the transport map $T_{\mu}^{-1}$ to each distribution $\mu$, and by using positive definite kernel on the Hilbert space $L^2(\bar{\mu})$, containing these transport maps. The following proposition provides the explicit kernel construction, and proves the positive definiteness.

\begin{proposition} \label{prop:embedding}
Consider a function $F : \mathbb{R}^+ \to \mathbb{R}$ such that, for any Hilbert space $H$ with norm $\| \cdot \|_H$, the function $h_1,h_2 \to  F(  \| h_1 - h_2 \|_H)$ is positive definite on $H$. 
Let $\bar{\mu}$ be a continuous distribution in $\mathcal{W}_2( \mathbb{R}^p )$. 
Consider the function $K$ on the set of continuous distributions in $\mathcal{W}_2( \mathbb{R}^p )$ defined by
\[
K (\mu , \nu )
= 
F( \|  T_{\mu}^{-1} - T_{\nu}^{-1} \|_{L^2(\bar{\mu})}  ).
\]
Then $K$ is positive definite. 
\end{proposition}
\begin{proof}
We use the following classical mapping argument.
For any $\lambda_1,\ldots,\lambda_n \in \mathbb{R}$ and continuous distributions $\mu_1,\ldots \mu_n$,
\[
\sum_{i,j=1}^n \lambda_i \lambda_j K( \mu_i , \mu_j )
=
\sum_{i,j=1}^n \lambda_i \lambda_j F( \| T_{\mu_i}^{-1} - T_{\mu_j}^{-1} \|_{L^2(\bar{\mu}} )
\geq 0
\]
because $F( \| \cdot \|_{L^2(\bar{\mu}})$ is positive definite on the Hilbert space $L^2(\bar{\mu})$.
\end{proof}

In Section \ref{s:properties}, we characterise all the continuous functions $F$ that satisfy the condition in Proposition \ref{prop:embedding}. Specific examples that can readily be used in practice are provided in \eqref{eq:square:exp} to \eqref{eq:power:exp}.

\begin{Rem} \label{rem:approxi:transport}
Proposition \ref{prop:embedding} will still hold, even if $T_{\mu}^{-1}$ is not exactly the inverse of an optimal transport map. The only constraint for Proposition \ref{prop:embedding} to hold is that $T_{\mu}^{-1}$ is uniquely defined as a function of $\mu$. Hence, in practice, we can use approximated optimal transport maps, and retain the positive definiteness guarantee (see also Section \ref{s:num}).
\end{Rem}

When the distributions $\mu_1,\ldots,\mu_n$ are observed, we recommend to select their empirical barcycenter as the reference distribution $\bar{\mu}$. If these distributions are realizations from a distribution  $\mathbb{P}  \in \mathcal{W}_2(\mathcal{W}_2(\mathbb{R}^p))$, the barycenter of $\mathbb{P}$ is also a good choice of a reference distribution, from a theoretical point of view.

\section{Gaussian processes indexed on Hilbert spaces} \label{s:properties}

We consider a real Hilbert space $H$ with inner product $( \cdot ,  \cdot )_H$ and norm $\|\cdot \|_H$.
In this section, we first characterize positive definite and strictly positive definite kernels, that are radial functions on $H$, that is functions of the form $F( \| \cdot - \cdot \|_H)$.
In Propositions \ref{prop:preservation:positive:definite} and \ref{prop:preservation:strictly:positive:definite}, we show that $F( \| \cdot - \cdot \|_H)$ is a (strictly) positive definite kernel on any Hilbert space $H$, if and only if it is a (strictly) positive definite kernel when $H = \mathbb{R}^d$ for any $d \in \mathbb{N}$. Thanks to these results, 
in Proposition \ref{prop:characterization:kernels}, we revisit classical results on radial positive definite functions on $\mathbb{R}^d$ \citet{wendland2004scattered}, by showing that when $F$ is continuous, $F( \| \cdot - \cdot \|_H)$ is strictly positive definite if and only if $F(\sqrt{\cdot})$ is completely monotone if and only if $F$ is an integral of negative square exponential functions with respect to a finite measure.

Second, we show in Theorem \ref{theo:microergo:hilbert} that when $H$ is of infinite dimension, virtually all covariance parameters are microergodic when considering Gaussian processes on bounded sets.

\subsection{Characterization of radial positive definite kernels}

We consider kernels $K: H \times H \to \mathbb{R}$ of the form
\begin{equation} \label{eq:radial}
K( u , v ) = F( \| u-v \|_H ),
\end{equation}
for $u,v \in H$. We call them radial kernels. The next proposition shows that $F$ provides a positive definite kernel on any Hilbert space $H$ if and only if it does so on finite dimensional Euclidean spaces.

\begin{Prop} \label{prop:preservation:positive:definite}
Let
$F :\mathbb{R}^{+} \rightarrow \mathbb{R}$.
Then the two following statements are equivalent.
\begin{enumerate}
\item For any $d \in \mathbb{N}$, the function $K_d : \mathbb{R}^d \times \mathbb{R}^d \to \mathbb{R} $ defined by $K_d(x,y) = F(\| x-y \|)$ for $x,y \in \mathbb{R}^d$ is positive definite.
\item For any Hilbert space $H$, the function $K$ of the form \eqref{eq:radial} is positive definite. 
\end{enumerate}
\end{Prop}

Next, we provide a similar characterization of the strict positive definitness property.

\begin{Prop} \label{prop:preservation:strictly:positive:definite}
Let
$F :\mathbb{R}^{+} \rightarrow \mathbb{R}$.
Then the two following statements are equivalent.
\begin{enumerate}
\item For any $d \in \mathbb{N}$, the function $K_d : \mathbb{R}^d \times \mathbb{R}^d \to \mathbb{R} $ defined by $K_d(x,y) = F(\| x-y \|)$ for $x,y \in \mathbb{R}^d$ is strictly positive definite.
\item For any Hilbert space $H$, the function $K$ of the form \eqref{eq:radial} is strictly positive definite. 
\end{enumerate}
\end{Prop}

In the case where $F$ is continuous, we can use the existing work on radial kernels on $\mathbb{R}^d$ (see e.g. \citet{wendland2004scattered}) to further characterize the functions $F$ providing strictly positive definite kernels in \eqref{eq:radial}. 
In this view, we call a function $f : [0,\infty) \to \mathbb{R}$ completely monotone if it is $C^{\infty}$ on $(0,\infty)$, continuous at $0$ and satisfies $(-1)^\ell f^{(\ell)}(r) \geq 0$ for $r > 0$.

\begin{Prop} \label{prop:characterization:kernels}
Let
$F :\mathbb{R}^{+} \rightarrow \mathbb{R}$.
Then the following statements are equivalent.
\begin{enumerate}
\item For any Hilbert space $H$, the function $K: H \times H \to \mathbb{R}$ of the form \eqref{eq:radial}, defined by $K( u , v ) = F( \| u-v \|_H )$, is strictly positive definite. 
\item $F(\sqrt{.})$ is completely monotone on $[0,\infty$ and not constant.
\item There exists a finite nonnegative Borel measure $\nu$ on $[0,\infty)$ that is not concentrated at zero, such that
\[
F(t) = \int_{\mathbb{R}} e^{- u t^2} \nu(du).
\]
\end{enumerate}
\end{Prop}
\begin{proof}
The proposition is a direct consequence of Proposition \ref{prop:preservation:strictly:positive:definite} and Theorem 7.14 in \citet{wendland2004scattered}. We remark that the statement 2. corresponds to a theorem from Schoenberg \citet{schoenberg1938metric}.
\end{proof}

From the previous proposition, it follows that the following choices of $F$ can be used in \eqref{eq:radial} to provide strictly positive definite covariance functions on $H$. The square exponential covariance function is given by
\begin{equation} \label{eq:square:exp}
F_{\sigma^2 , \ell}(t) = \sigma^2 e^{- (t/\ell)^2},
\end{equation}
with $\sigma^2,\ell \in (0,\infty)$.
The Mat\'ern covariance function is given by 
\begin{equation} \label{eq:matern}
F_{\sigma,\alpha,\nu}(t) = \frac{\sigma^2 ( \alpha t )^\nu }{2^{\nu -1} \Gamma(\nu)}
K_{\nu}( \alpha t )
\end{equation}
where $\Gamma$ is the Gamma function and $K_{\nu}$ is the modified Bessel function of the second kind \citet{stein99interpolation,loh2015estimating}.  
Finally, the power exponential function
\begin{equation} \label{eq:power:exp}
F_{\sigma^2,\ell,s}( t ) = \sigma^2 \exp( - (t/\ell)^s)
\end{equation}
satisfies the condition of Proposition \ref{prop:characterization:kernels} (see e.g. \citet{myieee}). 

Let us remark that, of course, not all positive definite kernels on $H$ are radial functions of the form \eqref{eq:radial}. For instance, the function $(\cdot , \cdot)_H$ is positive definite and is called a linear kernel.

One can also remark that, while Mercer's theorem has become classic for continuous positive definite kernels on compact sets of $\mathbb{R}^d$ \citet{wendland2004scattered}, a similar construction has not been shown to exist on bounded subsets of Hilbert spaces in infinite dimension. This can be considered as a structural difficulty when tackling Gaussian processes on infinite dimensional Hilbert spaces. On the other hand, we now show that infinite dimensional Hilbert spaces provide more space, so to speak, that enable to distinguish between distinct covariance functions in a more stringent way. More precisely, we show next that, when considering parametric sets of covariance functions, virtually all the covariance parameters are microergodic.

\subsection{Microergodicity results}

Let $ H$ be a Hilbert space.
Consider a set of functions $\{ F_{\theta} ; \theta \in \Theta \}$, with $F_{\theta} : \mathbb{R}^+ \to \mathbb{R}$ for $\theta \in \Theta$ and with $\Theta \subset \mathbb{R}^q$.  To $F_{\theta}$ we associate the covariance function $K_{\theta} = F_{\theta}(\| \cdot - \cdot \|)$ on $H$.

Let $h_0 \in  H$ and $0 < L < \infty$ be fixed and let $\bar{\mathcal{B}}_{2,L} = \{ h \in  H ; || h - h_0 ||_{ H} \leq L \}$. 
Let $ \bar{F} = \mathbb{R}^{ \bar{\mathcal{B}}_{2,L}}$ be the set of functions from $\bar{\mathcal{B}}_{2,L}$ to $\mathbb{R}$. Let $\mathcal{F}$ be the cylinder sigma algebra on $\bar{F}$ generated by the functions $f \to (f(h_1),...,f(h_r))$ for any $r \in \mathbb{N}$ and $h_1,...,h_r \in  H$. For any $\theta \in \Theta$, let $\mathbb{P}_{\theta}$ be the measure on $(\bar{F} , \mathcal{F})$ equal to the law of a Gaussian process on $\bar{\mathcal{B}}_{2,L}$ with mean function zero and covariance function $ (h_1,h_2) \to K_{\theta}( || h_1 - h_2 ||_{H} )$. Then, following \citet{stein99interpolation}, we say that the covariance parameter $\theta$ is microergodic if, for any $\theta_1, \theta_2 \in \Theta $ with $\theta_1 \neq \theta_2$, the measures $\mathbb{P}_{\theta_1}$ and $\mathbb{P}_{\theta_2}$ are orthogonal, that is there exists $\mathcal{A} \in \mathcal{F} $ so that $\mathbb{P}_{\theta_1}(A) = 1$ and $\mathbb{P}_{\theta_2}(A) = 0$. 

In the most classical case where $H = \mathbb{R}^d$, microergodicity is an important concept. Indeed, it is a necessary condition for consistent estimators of $\theta$ to exist under fixed-domain asymptotics \citet{stein99interpolation}, and a fair amount of work has been devoted to showing microergodicity or non-microergodicity of parameters, for various models of covariance functions \citet{stein99interpolation,zhang04inconsistent,And2010}. Typically, when $H = \mathbb{R}^d$ there are several standard sets of functions $\{ F_{\theta} ; \theta \in \Theta \}$ for which $\theta$ is not microergodic. A classical example is the set $\{ F_{\sigma^2 , \ell , \nu} \}$ of the form \eqref{eq:matern} \citet{zhang04inconsistent}.

In contrast, we now show that, under very mild assumptions, all covariance parameters $\theta$ are microergodic when $H$ has infinite dimension.

\begin{theorem}  \label{theo:microergo:hilbert}
Assume that $H$ has infinite dimension.
Assume that there does not exist $\theta_1 , \theta_2 \in \Theta$, with $\theta_1 \neq \theta_2$, so that $t \to F_{\theta_1}(t) - F_{\theta_2}(t)$ is constant on $[0,2L]$. Then the covariance parameter $\theta$ is microergodic. 
\end{theorem}

In Theorem \ref{theo:microergo:hilbert}, the condition on the parametric family $\{ F_{\theta} ; \theta \in \Theta \}$ holds for all the commonly used families of functions $F_{\theta}$ that are used to construct covariance functions on $\mathbb{R}^d$ as in Proposition \ref{prop:preservation:strictly:positive:definite}.
These commonly used families include notably the Mat\'ern covariance functions and the power exponential covariance functions that are introduced above. They also include
the generalized Wendland covariance functions and the spherical covariance functions  \citet{bevilacqua2016estimation,abrahamsen97review}. 

Hence, Theorem \ref{theo:microergo:hilbert} shows that it is possible that consistent estimators exist for $\theta$, in many parametric models of covariance functions of the form \eqref{eq:radial}, for infinite dimensional Hilbert spaces. 

Finally, one can see that if $\theta$ is microergodic when $H = \mathbb{R}^{d_2}$, then it is also microergodic when $H = \mathbb{R}^{d_1}$ with $d_1 \leq d_2$. That is, an higher dimension of the input space yields more microergodicity.
In agreement with this fact, Theorem \ref{theo:microergo:hilbert} can be interpreted as follows: when $d$ is infinite, the covariance parameter $\theta$ is always microergodic for Gaussian processes on $\mathbb{R}^d$.

\section{Statistical properties of our suggested positive definite kernels on distributions} \label{s:comput}

\subsection{General consistency properties}

Here, we consider the case where $n$ i.i.d. random  continuous distributions $\mu_1 , \ldots , \mu_n$ are observed, from a distribution $\mathbb{P}  \in \mathcal{W}_2(\mathcal{W}_2(\mathbb{R}^p))$. Hence, two possible reference distributions for our suggested construction of Proposition \ref{prop:embedding} are the empirical barycenter $\bar{\mu}_n$ of $\mu_1 , \ldots , \mu_n$ and the barycenter $\bar{\mu}$ of $\mathbb{P}$. We now show that these two reference points will asymptotically give the same kernel when $n$ is large.

For $\mu \in \mathcal{W}_2( \mathbb{R}^p )$, let $T_{\mu} , T_{\mu,n} : \mathbb{R}^p \to \mathbb{R}^p$ be the optimal transportation maps defined by 
\[
{T_{\mu}}_\sharp {\mu} = \bar{\mu}
~ ~
,
~ ~
{T_{\mu ,n}}_\sharp {\mu} = \bar{\mu}_n
\]
and
\[
|| \mathrm{id}  - T_\mu ||_{ L^2( \mu ) } = W_2( \mu , \bar{\mu} )
~ ~
,
~ ~
|| \mathrm{id}  - T_{\mu , n} ||_{ L^2( \mu ) } = W_2( \mu , \bar{\mu}_n ).
\]
Let also, for $i=1,...,n$ $T_i = T_{\mu_i}$ and $T_{i,n} = T_{\mu_i,n}$.

We remark that, because of the assumption on $\mathbb{P}$, both the barycenter and the empirical barycenter are absolutely continuous w.r.t Lebesgue measure on $\mathbb{R}^p$. Hence, $T_1,...,T_n$ and $T_{1,n},...,T_{n,n}$ are uniquely defined. For $F : \mathbb{R}^+ \to \mathbb{R}$, we let
\begin{equation} \label{eq:Kn}
K_{n}(\mu,\nu) = F (\|T_{\mu,n}^{-1}-T_{\nu,n}^{-1}\|^2_{L^2(\bar{\mu}_n)})
\end{equation}
be the  empirical kernel and
\begin{equation} \label{eq:K}
K(\mu,\nu) = F (\|T_{\mu}^{-1}-T_{\nu}^{-1}\|^2_{L^2(\bar{\mu})})
\end{equation}
be the theoretical kernel.
We now prove that the empirical kernel $K_n$ provides a good approximation of the kernel $K$. 
We will use the consistency property of Theorem~\ref{th:empiricconsisten}, stating that the empirical barycenter is a consistent estimate for $\bar{\mu}$. 

\begin{Prop}[Consistency of Kernel] 
\label{prop:approximation}
Let $F$ in \eqref{eq:Kn} and \eqref{eq:K} be continuous. 
The empirical kernel is a good approximation of the true covariance kernel in the sense that, for any two fixed absolutely continuous measures $\mu$ and $\nu$ in $\mathcal{W}_2(\mathbb{R}^p)$, we have 
$$
K_{n}(\mu,\nu) \rightarrow K(\mu,\nu)
$$
a.s. when $n$ goes to infinity.
\end{Prop}
\begin{proof}
Using the continuity of the function $F$, it is enough to show that a.s. 
\[
\|T_{\mu,n}^{-1}-T_{\nu,n}^{-1}\|^2_{L^2(\mu_n)} 
-\|T_{\mu}^{-1}-T_{\nu}^{-1}\|^2_{L^2(\bar{\mu})} 
\longrightarrow 0.
\]
Lemma~\ref{lem:true}, whose proof is presented in the Appendix, leads to the result.
\end{proof}

In the next Corollary, we show that the consistency result in Proposition \ref{prop:approximation} implies that the conditional means and variances based on the empirical kernel asymptotically coincide with those based on the true kernel.

\begin{corollary} \label{cor:consistency:inference}
Let $N \in \mathbb{N}$ and let $\mu_1,\ldots,\mu_N,\mu$ be fixed absolutely continuous measures in $\mathcal{W}_2(\mathbb{R}^p)$. Let $y= (y_1,\ldots,y_N)^\top$ be fixed in $\mathbb{R}^N$. Set $R = [K(\mu_i , \mu_j)]_{1 \leq i,j \leq N}$ and assume that $R$ is invertible. Let $Y = \{Y_{\mu}\}$ be a Gaussian process with zero mean function and covariance function given by \eqref{eq:K}. Then $$\mathbb{E}( Y_\mu | Y_{\mu_1} = y_1 , \ldots ,  Y_{\mu_N} = y_N ) = r_{\mu}^\top R^{-1} y$$ with $r_{\mu} = ( K(\mu,\mu_1),\ldots,\linebreak[1] K(\mu,\mu_N) )^\top$.  Let $$\mathbb{E}_n( Y_\mu | Y_{\mu_1} , \ldots ,  Y_{\mu_N} ) = r_{\mu,n}^\top R_n^{-1} y$$ with $r_{\mu,n} = ( K_n(\mu,\mu_1),\ldots, K_n(\mu,\mu_N) )^\top$ and $R_n = [K_n(\mu_i , \mu_j)]_{1 \leq i,j \leq N}$. Also $$\mathrm{Var}( Y_\mu | Y_{\mu_1} = y_1, \ldots ,\linebreak[1]  Y_{\mu_N} = y_N) = K(\mu,\mu) -  r_{\mu}^\top R^{-1} r_{\mu}$$ and we let $$\mathrm{Var}_n( Y_\mu | Y_{\mu_1} , \ldots ,  Y_{\mu_N} ) = K_n(\mu,\mu) -  r_{\mu,n}^\top R_n^{-1} r_{\mu,n}.$$
\\
Then, a.s. as $n \to \infty$, \begin{align*} \mathbb{E}_n( Y_\mu | Y_{\mu_1} , \ldots ,  Y_{\mu_N} ) & \to \mathbb{E}( Y_\mu | Y_{\mu_1} , \ldots ,  Y_{\mu_N})  \\  {\rm and} \:  \mathrm{Var}_n( Y_\mu | Y_{\mu_1} , \ldots ,  Y_{\mu_N} ) & \to \mathrm{Var}( Y_\mu | Y_{\mu_1} , \ldots ,  Y_{\mu_N} ).\end{align*}
\end{corollary}

\begin{proof}
The Corollary is a direct consequence of the facts that $N$ is fixed as $n \to \infty$ and that $R$ is invertible.
\end{proof}


\subsection{Universality}

Note that when considering a kernel $K$, a natural property to be studied would be its universality. Actually, a kernel is said to be universal on $\Omega \subset \mathcal{W}(\mathbb{R}^p)$ as soon as the space generated by its linear combinations $\mu \in \mathcal{W}(\Omega) \mapsto \sum_{i=1}^n \alpha_i K(\mu,\mu_i) \in  \mathbb{R}$ can generate all continuous functions on $\mathcal{W}(\Omega)$. The general form \eqref{eq:K} of the kernel may provide uniform kernels under regularity assumptions on the transportation maps $T_i$. More precisely injectivity and continuity are required as pointed out in \citet{micchelli2006universal} to get a universal kernel. In some particular cases, it is possible to obtain such results. In the case of Gaussian distributions, the transport map is linear and thus it entails the universality of the kernel in this case. In \citet{2018arXiv180601238D}, Proposition 1.4.1 derived Theorem 1.1 from \citet{2018arXiv180504946F} provides some conditions for continuity of the transportation maps but regularity of transportation maps in general dimensions is a difficult issue. It  has received a lot of attention in the last years see  for instance to~\citet{santambrogio2015optimal} and such conditions can not be guaranteed in a very general framework but could only be studied for very particular class of distributions, leading to too restrictive cases, which are not at the heart of this paper. 

\subsection{Specific properties for Gaussian distributions}

In some special cases, the optimal transportation maps can be written down explicitly. Unfortunately, this holds only for some particular class of admissible transformations. 
An example of explicit calculations is given by 
a  family of  Gaussian distribution.
Let \(\mathcal{F} = \left\{\mathcal{N}(0, S)\right\}_{S} \)
be a family of centred Gaussian distributions. 
Further we assume the covariance matrices to be random:
\(S \overset{iid}{\sim} \mathbb{P} \).
This setting is equivalent to the definition of
some distribution \(\mathbb{P}\) over \(\mathcal{F}\).
We denote as $\mub = \mathcal{N}(0, \Sb)$
the unique population barycenter of \(\mathbb{P}\).

Let \(\{\mu_i\}_{i = {1,\dots,n}}\) be a family of observed 
random Gaussian distributions with  zero mean and
non-degenerated covariance \( S_i\): \(\mu_i = \mathcal{N}(0, S_i) \),
\(S_i \sim \mathbb{P} \).
An empirical barycenter is recovered uniquely: \(\mun = \mathcal{N}(0, \Sb_n) \)
with \(\Sb_n\) a solution of the following fixed-point equation 
\(\Sb_n = \frac{1}{n}\sum\bigl(S^{1/2}_i\Sb_n S^{1/2}_i \bigr)^{1/2} \). This result is well known and has been described in many papers, see for instance in the seminal work \citet{agueh2011barycenters}. The solution can be obtained by an iterative method, presented in \citet{ALVAREZESTEBAN2016744}.

The Gaussian setting allows to write down an optimal transport plan \(T_i \) between 
$\mu_i$ and the population barycenter $\mub = \mathcal{N}(0,\bar{S})$ and its inverse explicitly:
\[
\label{eq:GauOT}
T_i =  S^{-1/2}_i\bigl(S^{1/2}_i \Sb S^{1/2}_i \bigr)^{1/2}S^{-1/2}_i,
\quad
T_i^{-1} = {\Sb}^{-1/2}\bigl({\Sb}^{1/2} S_i {\Sb}^{1/2} \bigr)^{1/2} {\Sb}^{-1/2}.
\]
In this case, we can compute the distance between the transport plans in $L^2(\bar{\mu})$ using the expression in \eqref{eq:GauOT}
$ \|T_i^{-1}-T_j^{-1}\|^2_{L^2(\bar{\mu})}$, as the distance is the variance of a linear transform of Gaussian random variable:
\begin{equation}
\label{eq:OT_map_norm}
\|T_i^{-1}-T_j^{-1}\|^2_{L^2(\bar{\mu})} 
= \Bigl\|\Sb^{-1/2} \Bigl[ \bigl({\Sb}^{1/2} S_i {\Sb}^{1/2} \bigr)^{1/2} 
- \bigl({\Sb}^{1/2} S_j {\Sb}^{1/2} \bigr)^{1/2} \Bigr]^{1/2} \Bigr\|^2_{F}.
\end{equation}
The same expression holds for \(\bigl\|T_{i, n}^{-1}-T_{j, n}^{-1} \bigr\|^2_{L^2(\bar{\mu}_{{n}})} \), replacing the barycenter by its empirical counterpart. We can see that in this case the kernel amounts to compute a natural distance between the two distributions $\mu_i$ and $\mu_j$ obtained by the scale deformation $S_i^{1/2} X$ and $S_i^{1/2} X$ of a Gaussian random variable $X\sim \mathcal{N}(0, Id)$. This distance is then used through any kernel which provides some insights on a proper notion of covariance between processes indexed by these two distributions.

\indent We point out that in the Gaussian case, the rate of convergence of the covariance estimates can be made precise.
\begin{Prop} \label{p:rate}
	Let  \(\mathcal{F}\) be s.t. 
	\(\mathbb{E}_{S\sim \mathbb{P}}\text{tr}(S) \leq +\infty\) and
	let $M_n$ and $M$ be respectively the empirical and true $N \times N $ covariance matrices 
	of a Gaussian process constructed from the kernels \(K_n\) and \(K\) using a grid $\mathcal{N}(0, S_1),\dots,\mathcal{N}(0, S_N)$, \(S_i \sim \P\), defined as in \eqref{eq:Kn} and \eqref{eq:K}.
	Then there exists a finite constant $C$ such that with high probability  
	$$ \|M_n-M\|^2_F \leq C \frac{N^2}{n}.$$
\end{Prop}

Finally, for the Kernel with Gaussian distributions, it is possible to understand the stationarity property of the kernel. The following proposition illustrates that in the Gaussian case the kernel is indeed invariant with respect to orthogonal transformations. 

\begin{Prop}
	\label{prop:stationarity}
	Let \(U\) be some predefined orthogonal matrix, and set set \(\phi_{U}\) be a deterministic map, that sends any \(\mathcal{N}(0,S)\)to \(\mathcal{N}(0,USU^T)\). For any \(i = 1,..., n\) denote as
	\(T_{i, \phi}\) the optimal transportation map \( T_{i, \phi \sharp} \phi_{U}\left(\mathcal{N}(0, S_i)\right) 
	= \phi_{U}\left(\mathcal{N}(0, \Sb) \right)\). Then it holds
	\begin{equation}
	\label{eq:stationarity}
	\left\|T^{-1}_{i, \phi} - T^{-1}_{j, \phi} \right\|_{L^2(\phi_U(\mub))} 
	= \left\|T^{-1}_i - T^{-1}_j \right\|_{L^2(\mub)}.
	\end{equation}
\end{Prop}
Equality~\eqref{eq:stationarity} ensures stationarity of the kernels under application of transformation \(\phi_{U}\).

\section{Numerical simulations}\label{s:num}

\subsection{Computational aspects}

In practice, finding analytical representations of optimal transportation maps
is a difficult issue, 
especially if the dimension of the problem grows.
A possible solution consists in approximating 
an optimal transportation map by its empirical counterpart.
Let \(\mu_{m}\) and \(\nu_{m} \) be empirical measures 
sampled from \(\mu \) and \(\nu \) respectively. Then the optimal Monge map 
\(T_{\sharp}\mu = \nu \) can be replaced by \(T^{m}_{\sharp}\mu_m = \nu_m \), see e.g.
~\citet{chernozhukov2017monge} or~\citet{boeckel2018multivariate}.
In this case, the problem of finding \(T^{m}\) is reduced to the solution of
assignment problem with quadratic cost and can be solved by the \textit{adagio} \textsf{R}-package by  ~\citet{borchersadagio}.

In dimension $p=2$ or $p=3$, it is also possible to represent the distributions by their matrices of probability weights on regular grids. Optimal transport maps can then be approximated, by means of various numerical procedures \citet{luenberger1984linear,gottschlich2014shortlist,merigot2011multiscale}. In our practical implementations, we tend to use the packages \citet{trans} and \citet{barycenter}, with the \textsf{R} programming language.

\subsection{ Numerical study of the kernel consistency on a subspace of Gaussian measures}
In what follows we present some simulations to highlight the consistency of the empirical kernel obtained in the Gaussian case from the empirical barycenter. 
For this we consider a population \(\mathcal{F}\) of \(100000\) centred Gaussians on \(\mathbb{R}^d\) with covariance \(S_i = A_iA'_i \), 
with \(i = 1,..., n\),
where 
\( A_i = (a_{jk})^d_{j, k = 1}, a_{jk} \sim ~\text{Unif}[5, 15] \). In these experiments we consider \(d = (4, 7, 15, 30)\).
We compute the true barycenter \(\mathcal{N}(0, \Sb)\) for which the whole \(\mathcal{F}\) is used, while \(\Sb_n\) is computed as a Wasserstein-mean of a random \(n\)-sample (with replacement) from \(\mathcal{F}\).
Let \(M\) and \(M_n\) be the covariance matrices, 
obtained from the kernels \(K\) and \(K_n\) constructed using~\eqref{eq:square:exp} with parameters \(l = \sigma = 1\) on a grid of \(N = 30\) randomly selected measures from \(\mathcal{F}\).

Table~\ref{tab:gauss} illustrates the mean approximation error rate \(\|M_n - M \|_{F}\) for the cases \( n = (20, 140, 260, 380, 500, 620)\).

\begin{table}[htbp]
	\label{tab:gauss}
	{\caption{Error: \(\|K_n - K\|_{F}\) for centred Gaussians on \(\mathbb{R}^d\)}}
	{%
		\begin{tabular}{l|llllll}
			& n =  20 & n = 140 & n = 260 & n = 380 & n = 500 & n = 620\\ \hline
			d = 4   & 1.52 & 0.69 & 0.16 & 0.29 & 0.24 & 0.14 \\ 
			d = 7 & 2.08 & 0.59 & 0.17 & 0.19 & 0.11 & 0.14  \\
			d = 15 & 0.91 & 0.12 & 0.09 & 0.08 & 0.05 & 0.05 \\
			d = 30  & 0.90 & 0.13 & 0.05 & 0.03 & 0.04 & 0.02
		\end{tabular}
	}
\end{table}

As expected, we can see convergence of the  empirical kernel towards the theoretical one in all cases. 

\subsection{Prediction experiments on simulated data} \label{subsection:pred:gaussians}

Then we consider the following simulations for the 2 dimensional case. We simulate 100 random two-dimensional Gaussian distributions split into a training sample of 50 and a test sample of 50. Both mean vectors and covariance matrices are chosen randomly. The mean vector follows a uniform distribution over $[0.2,0.8]^2$. The covariance matrix is isotropic and the standard deviation is uniform over $[0.01^2,0.02^2]$. The value of the random field $Y$ for a Gaussian distribution $\mu$, given by its mean $(m_1,m_2)^T$ and variance $\sigma^2$, is given by 
$$Y(\mu)= \frac{(m_1-m_2^2)}{ 1+\sigma}.$$

We then carry out our suggested Gaussian process model, based on the kernels suggested in Proposition \ref{prop:embedding}.
Optimal transport maps $T_{\mu}^{-1}$, from the barycenter to the Gaussian measures $\mu$, are calculated using the package  \citet{trans} and barycenters are calculated using the package \citet{barycenter} with parameter $\lambda=20$ to balance computational time and similarity between the penalized transport and the optimal transport without regularization. \\

More precisely, the Gaussian distributions are discretized over a grid of $50 \times 50$ cells on  $[0,1]^2$. The Gaussian distributions are thus approximated by discrete distributions on the grid. We remark that the package \citet{trans} does not exactly provide deterministic transport maps. Indeed, the probability mass of a given input grid point can be split and mapped to several output grid points. Numerically, in this case, we transport all the probability mass of the input grid point to the output grid point that is assigned the most mass by the package \citet{trans}. Hence, to each discretized input Gaussian measure $\mu$, we associate a transport map $T_{\mu}^{-1}$ from the barycenter that is an approximation of the inverse of the optimal transport map from $\mu$ to the barycenter. Nevertheless, since the mapping from $\mu$ to $T_{\mu}^{-1}$ is uniquely defined in our procedure, Remark \ref{rem:approxi:transport} applies and we are guaranteed to obtain positive definite kernels.

The kernel we choose is $K_\theta$ given by
$$K_{\theta}(\mu,\nu):=\theta_1^2 *\exp(-\theta_2 \left\|T^{-1}_{\mu} - T^{-1}_{\nu} \right\|_{L^2(\mub)}^{\theta_3})+ \theta_4 1_{\left\|T^{-1}_{\mu} - T^{-1}_{\nu} \right\|_{L^2(\mub)}=0}$$ 
for $\theta_1 \in [0.05,10]$, $\theta_2 \in [0.01,10]$, $\theta_3 \in [0.5,2]$ and $\theta_4 \in [10^{-5},1]$. We will use the kernel with the parameters chosen to maximize the likelihood but also parameters chosen to minimize the sum of the cross-validation square errors \citet{bachoc2013cross,bachoc2018asymptotic}. For cross-validation, the total variance parameter $\theta_1^2 + \theta_4$ is estimated as suggested in \citet{bachoc2013cross}.

We compare our kernel methods with the kernel smoothing procedure of \citet{pmlr-v31-poczos13a}. 
This procedure consists in predicting $Y(\mu) \in \mathbb{R}$ by a weighted average of $Y(\mu_1),...,Y(\mu_n)$ where the weights are computed by applying a kernel to the distances $D(\mu,\mu_1),...,D(\mu,\mu_n)$ where $D$, as suggested in \citet{pmlr-v31-poczos13a} is the $L^1$ distances between the  probability density functions. The kernel is the triangular kernel as in \citet{pmlr-v31-poczos13a}, and its bandwidth is selected by minimizing an empirical mean square error based on sample splitting (see \citet{pmlr-v31-poczos13a}).
We remark that there is no estimate of the prediction error $Y(\mu) - \hat{Y}(\mu)$  which is a downside compared to the Gaussian process model considered in this paper.\\ 

We present hereafter in Table~2 the results obtained, with $50$ observations and $50$ values to be predicted. We study the Root Mean Square Error (RMSE) of the form
\[
\sqrt{
	\frac{1}{50} \sum_{i=1}^{50} ( \hat{Y}_i - Y_i )
},
\]
where the $Y_i$ are the values to be predicted and the $\hat{Y}_i$ are the predictions. We also study the $Q^2$ criterion which is equal to $1 - \text{RMSE}^2 / \text{var}$, where $\text{var}$ is the empirical variance of the values to be predicted. Finally we study the Confidence Interval Coverage (CIC) which corresponds to the frequency of the event that the predicted value belongs to the $90\%$ confidence interval from the Gaussian process model.
\begin{table}[htbp]
	{\caption{Prediction results for Gaussian simulations.}}
	{%
		\begin{tabular}{l|lll}
			& RMSE & $Q^2$ & CI Coverage  \\ \hline
			Kernel Smoothing   & 0.15  & 0.61 & NA \\ 
			Gaussian Process  & 0.10 & 0.81 & 0.87 \\
			Gaussian Process CV  & 0.10 & 0.81 & 0.88
		\end{tabular}
		\label{GPresult}
	}
\end{table}
From the table, one observes that the GP process model based on the kernel we suggest provides a better accuracy, catching better the variability of the underlying process. 

\subsection{Experiments on real data : stress response to traction for materials in nuclear safety}
We focus on a computer code called CASTEM code (see \citet{CASTEM}) from the French Atomic Energy Commission (CEA) designed to calculate equivalent stresses on biphasic materials subjected to uni-axial traction. The system  is modelled as a unit square containing $m$ circular inclusions, all with the same radius $R$ at random locations associated to a numerical value which is the stress response. The simulations are performed in two dimensions over  $[0,1]^2$. The input of the codes are $m=10$ disks located at $m$ points $\{c_1,\dots,c_m\}$ while the stress responses are scalar numerical values provided by the CASTEM code.
As pointed out in \citet{ginsbourger2016design}, finding a proper distance between the inputs to forecast the stress is a very difficult task. \\

In this framework, we propose to consider each input as a uniform distribution $\mu$ on the union of the disks. For all the inputs $i=1,\dots,n$, we let $c^{(i)}=(c_1^{(i)},\dots,c_m^{(i)})$ be the vector of dimension $2m$ composed by the $m$ centers of the disks and we let $D_j^{(i)}$ be the disk with center $c_j^{(i)}$ and radius $R$. Then we let $\mu_i$ be the Uniform distribution over $\cup_{j=1}^m D^{(i)}_j$. Then the stress is  considered as a Gaussian random field indexed by the $\mu_i$'s. \\
As previously, to compute the barycenter, we use the package provided in \citet{barycenter}. We use a grid over $[0,1]^2$ that discretizes the set into $50 \times 50$ cells. The uniform distribution on the set of disks is evaluated onto these cells and is approximated by a discrete distribution that is considered as an image. The optimal transport maps from the distribution to the barycenters are calculated using \citet{trans}, similarly as in Section \ref{subsection:pred:gaussians}. We compare to the kernel smoothing procedure also as in Section \ref{subsection:pred:gaussians}.

The results are presented in Table 3 in the same way as in Table 2. In Table 3, the methods use $500$ ouputs of the CASTEM code and predict $400$ other outputs.

\begin{table}[htbp]
	{\caption{Prediction of the CASTEM code output.}}
	{
		\begin{tabular}{l|lll}
			& RMSE & $Q^2$ & CI Coverage \\\hline
			Kernel Smoothing  & 0.96      & 0.03   & NA \\
			Gaussian Process   & 0.93 & 0.10 & 0.92 \\
			Gaussian Process CV & 0.92 & 0.11 & 1
		\end{tabular}
	}
\end{table}
As noted by many specialists, forecasting the CASTEM code is a very hard task, given that the inputs are very complex, which explains the poor $Q^2$ score for the three methods. Yet the method proposed in this work provides some improvements with respect to the state of the art method from \citet{pmlr-v31-poczos13a}. We point out that cross validation of the parameters for the Gaussian Process provides a very small improvement of the prediction but at the expense of overly large confidence intervals. 

We remark that the kernel we provide is a positive definite kernel as required to use the Gaussian Process modelling framework. Using directly a kernel by computing the exponential of the square $W_2$ Wasserstein distance between the distributions does not lead to a positive definite kernel. Actually Figure \ref{fig:my_label} shows the repartition of the eigenvalues of the $900 \times 900$ covariance matrix based on this kernel. We observe that many eigenvalues are negative (before the red line in the figure where we plot the logarithm of $1$ plus the eigenvalues). 
\begin{figure}
	\centering
	\includegraphics[height=10cm,width=14cm]{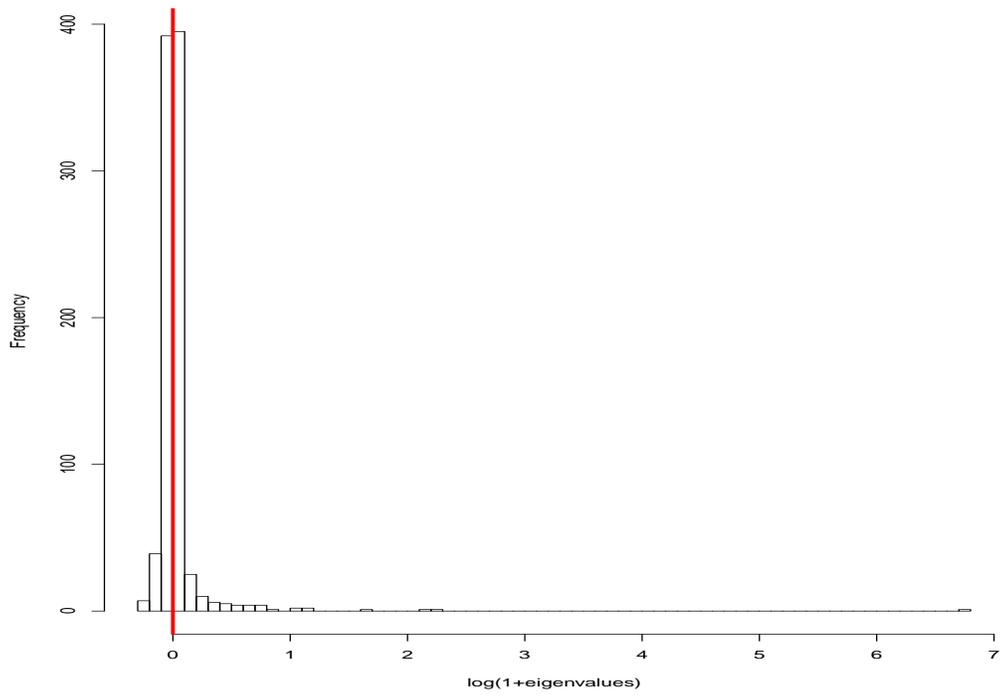}
	\caption{Distribution of the eigenvalues of the $900 \times 900$ matrix obtained by the kernel of the form $\exp( -W_2^2( \mu , \nu ) )$. Many eigenvalues are negative, which shows that this kernel is not positive definite.}
	\label{fig:my_label}
\end{figure}

\section{Conclusion and Future Directions} \label{s:4uSasha}
In this work, we have provided a theoretical way to use Wasserstein barycenters in order to define general kernels using optimal transportation maps. Considering the distance between the optimal transportation maps provide a natural way to quantify correlations between the values of a process indexed by the distribution and provides a generalization to multi-dimensional case of the work in~\citet{myieee}. Using barycenter requires that the distributions are drawn according to the same measure over the set of distributions. This restricts the framework of the study to the case where the Gaussian process is defined on the support of this measure. For applications, this does not a play a too important feature since inputs are often simulated according to a specified distribution. Yet for theoretical issues, this sets the frame of this study to the infill case and not the asymptotic frame. In this case, few results exist in the statistical literature on Kriging, and thus we focused on micro-ergodicity of the parameters, proving that consistent estimate can be studied.  \\
Finally contrary to the one-dimensional case, computational issues arise naturally when the  Wasserstein distance is required. Hence the computation of a barycenter with respect to Wasserstein distance is a difficult optimization program, unless the distributions are Gaussian, leading to tractable computations as shown in Section~\ref{s:comput}.  Yet this idea of linearization around the barycenter to obtain a valid covariance kernel could be used and generalized to regularized Wasserstein distance using methods proposed in \citet{cuturi2014fast} for instance  to provide a more tractable way of building kernels. 

\vskip .1in

\noindent
\textit{Acknowledgments: } the authors would like to thank Drs. Jean Baccou and Fr\'ed\'eric P\'erales (respectively at LIMAR and LPTM, Institut de Radioprotection et de S\^uret\'e Nucl\'eaire, Saint-Paul-l\`es-Durance, France) for the CASTEM data set, and Dr. Cl\'ement Chevalier (now with the Swiss Statistical Office, Neuch\^atel, Switzerland) who has been involved in investigations on this data set in the framework of the \href{http://redice.emse.fr/}{ReDICE consortium}.

\vskip 0.2in
\bibliographystyle{plainnat}
\bibliography{references.bib}
\newpage
\appendix

\section{Proofs}

\begin{description}
	\item[Proof of Propositions~\ref{prop:preservation:positive:definite} and \ref{prop:preservation:strictly:positive:definite}]
\end{description}
\begin{proof}
	For both propositions, only the fact that 1. implies 2. needs to be proved. Let us now do this.
	
	Let $f_1,\dots,f_n$ in $H$ and consider the matrix $\tilde{C}=((f_i,f_j)_H)_{\{i,j\}}$.  This matrix is a Gram matrix in $\mathbb{R}^{n \times n}$ hence there exists a non negative diagonal matrix $D$ and an orthogonal matrix $P$ such that
	$$ \tilde{C}=PDP^{'} = PD^{1/2} D^{1/2}P^{'}.$$
	Let $e_1,\dots,e_n$ be the canonical basis of $\mathbb{R}^n$. Then
	$$ e_i \tilde{C} e_j^{'}= u_i u_j^{'}$$ where $u_i=e_i P D^{1/2}.$ Note that the $u_i$'s are vectors in $\mathbb{R}^n$ that depend on the $f_1,\dots,f_n$. By polarization, we hence get that $ (f_i,f_j)_H=(u_i,u_j) $ where $(,)$  denotes the usual scalar product on $\mathbb{R}^n$. Hence we get that for any elements $f_1,\dots,f_n$ in $H$ there are $u_1,\dots,u_n$ in $\mathbb{R}^n$ such that $\|f_i-f_j\|_H=\|u_i-u_j\|.$ So any covariance matrix that can be written as $[ F(\|f_i-f_j\|_H) ]_{i,j}$ can be seen as a covariance matrix $[ F(\|u_i-u_j\|) ]_{i,j}$ on $\mathbb{R}^n$ and inherits its properties. The invertibility and non-negativity of this covariance matrix entail the invertibility and non-negativity of the first one, which proves the results.
\end{proof}

\begin{description}
	\item[Proof of Theorem~\ref{theo:microergo:hilbert}]
\end{description}
\begin{proof}
	Without loss of generality, we can assume that $h_0 = 0 \in  H$. 
	Let $\theta_1 , \theta_2 \in \Theta$, with $\theta_1 \neq \theta_2$. Then, there exists $t^* \in [0,L] $ so that $F_{\theta_1}(0) - F_{\theta_1}(2 t^*) \neq F_{ \theta_2 }(0) - F_{ \theta_2 }( 2t^*)$.
	
	For any $n \in \mathbb{N}$, let $e_1,...,e_n \in  H$ satisfy $( e_i , e_j )_{ H} = \mathbf{1}_{i=j}$.
	Consider the $2n$ elements $(f_1,...,f_{2n})$ made by the pairs $(-t^* e_i, t^* e_i)$ for $i=1,\dots,n$. Consider a Gaussian process $Y$ on $\bar{\mathcal{B}}_{2,L}$ with mean function zero and covariance function $K_{ \theta_1 }$. Then, the Gaussian vector $Z = (Y(f_i))_{i=1,...,2n}$ has covariance matrix $C$ given by 
	\[
	C_{i,j} =
	\begin{cases}
	F_{ \theta_1 }(0) & \mbox{if } i=j \\
	F_{ \theta_1 }(2 t^*) & \mbox{if $i$ even and $j=i+1$}  \\
	F_{ \theta_1 }(2 t^*) & \mbox{if $i$ odd and $j=i-1$}  \\
	F_{ \theta_1 }( \sqrt{2} t^*) & \mbox{else}.
	\end{cases}
	\]
	Hence, we have $C = D + M$ where $M$ is the matrix with all components equal to $K_{ \theta_1 }( \sqrt{2} t^*)$ and where $D$ is block diagonal, composed of $n$ blocks of size $2 \times 2$, with each block equal to 
	\[
	B_{2,2} = 
	\begin{pmatrix}
	F_{ \theta_1 }(0) - F_{ \theta_1 }( \sqrt{2} t^*) &  F_{ \theta_1 }(2 t^*) -  F_{ \theta_1 }( \sqrt{2} t^*) \\
	F_{ \theta_1 }(2 t^*) -  F_{ \theta_1 }( \sqrt{2} t^*) & 
	F_{ \theta_1 }(0) - F_{ \theta_1 }( \sqrt{2} t^*)
	\end{pmatrix}.
	\]
	Hence, in distribution, $Z = M + E$, with $M$ and $E$ independent, $M=(z,....,z)$ where $z \sim \mathcal{N}(0,K_{ \theta_1 }( \sqrt{2} t^*))$ and where the $n$ pairs $(E_{2k+1},E_{2k+2})$, $k=0,...,n-1$ are independent, with distribution $\mathcal{N}(0,B_{2,2})$. Hence, with $\bar{Z}_1 = (1/n) \sum_{k=0}^{n-1} Z_{2k+1}$, $\bar{Z}_2 = (1/n) \sum_{k=0}^{n-1} Z_{2k+2}$ and $\bar{E} = (1/n) \sum_{k=0}^{n-1} (E_{2k+1},E_{2k+2})^t$, we have
	\begin{align*}
	\hat{B} & :=
	\frac{1}{n} \sum_{i=0}^{n-1}
	\begin{pmatrix}
	Z_{2i+1} - \bar{Z}_1 \\
	Z_{2i+2} - \bar{Z}_2
	\end{pmatrix}
	\begin{pmatrix}
	Z_{2i+1} - \bar{Z}_1 \\
	Z_{2i+2} - \bar{Z}_2
	\end{pmatrix}^t  \\
	& = \frac{1}{n} \sum_{i=0}^{n-1}
	\begin{pmatrix}
	E_{2i+1}  \\
	E_{2i+2} 
	\end{pmatrix}
	\begin{pmatrix}
	E_{2i+1}  \\
	E_{2i+2} 
	\end{pmatrix}^t
	- \bar{E} \bar{E}^t \\
	& \to_{n \to \infty}^{p}
	B_{2,2}.
	\end{align*}
	Hence, there exists a subsequence $n' \to \infty$ so that, almost surely $\hat{B} \to B_{2,2}$ as $n' \to \infty$. Hence, almost surely $\hat{B}_{11} - \hat{B}_{1,2} \to K_{ \theta_1 }(0) - K_{ \theta_1 }(2 t^*)$ as $n' \to \infty$. Hence, the set 
	\[
	A = \left\{
	g \in \bar{F};
	\hat{B}_{2,2}
	\left(
	g(f_1),....g(f_{2n'})  
	\right) 
	\to_{n' \to \infty} 
	F_{ \theta_1 }(0) - F_{ \theta_1 }(2 t^*)
	\right\}
	\]
	satisfies $P_{\theta_1}(A) = 1$.
	With the same arguments, we can show $P_{\theta_2}(B) = 1$, where 
	\[
	B= \left\{
	g \in \bar{F};
	\hat{B}_{2,2}
	\left(
	g(f_1),....g(f_{2n''})  
	\right) 
	\to_{n'' \to \infty} 
	F_{ \theta_2 }(0) - F_{ \theta_2 }(2 t^*)
	\right\}
	\]
	where $n''$ is a subsequence extracted from $n'$. Since $A \cap B = \varnothing$, it follows that $P_{\theta_2}(A) = 0$. Hence, $\theta$ is microergodic.
	
\end{proof}

\begin{description}\item[{Proof of Proposition~\ref{prop:approximation}}]
\end{description}

Recall that the empirical barycenters 
\((\mun)_n\) is a sequence of
continuous measures converging to \(\mub\) in 
\(2\)-Wasserstein distance: \(W_2(\mun, \mub) \rightarrow 0 \)
as \(n \rightarrow \infty\) and \(R_{n\sharp}\mub = \mun\) with $W_2^2( \bar{\mu},\bar{\mu}_n ) = || R_n ||_{L^2(\bar{\mu})}$. 


\begin{lemma}
	\label{lemma:comvergence_OT}
	Fix some distribution  \(\nu\) absolutely continuous with respect to Lebesgue measure and let $T = T_{\nu}$ and $T_n = T_{\nu,n}$.
	Then it holds a.s.
	\[
	\bigl\|T - T_n \bigr\|^2_{L^2(\nu)} \longrightarrow 0, ~\text{as}~ n \rightarrow \infty.
	\]
\end{lemma}
\begin{proof}
	Fix \(n\) s.t. \(W_2(\mun, \mub) = \varepsilon_{n}\). 
	Consider \(\bigl\|\text{id} - R_n\circ T \bigr\|_{L^2(\nu)}\). 
	By change of variables and triangle inequality one obtains
	\begin{align*}
	\bigl\|\text{id} - R_n\circ T \bigr\|_{L^2(\nu)} 
	& = 
	\bigl\|T^{-1} - R_n \bigr\|_{L^2(\mub)}
	\leq 
	\bigl\|T^{-1} - \text{id} \bigr\|_{L^2(\mub)} 
	+ 
	\bigl\|R_n - \text{id} \bigr\|_{L^2(\mub)} \\
	& \leq W_2(\nu, \mub) +\varepsilon_n
	\leq W_2(\nu, \mun) +2\varepsilon_n.
	\end{align*}
	Since $T_n$ is the optimal transport map from $\nu$ to $\mu_n$ we recall that $W_2(\nu, \mun)= \bigl\|\text{id} - T_n \bigr\|_{L^2(\nu)} $. So  
	due to the arbitrary choice of \(n\) it follows
	\begin{equation}
	\label{eq:convergence}
	\Bigl| \bigl\|\text{id} - R_n\circ T\bigr\|_{L^2(\nu)} - \bigl\|\text{id} - T_n \bigr\|_{L^2(\nu)}  \Bigr| 
	\underset{n \rightarrow \infty}{\longrightarrow} 0.
	\end{equation}
	Now we are ready to prove, that \(\bigl\|T_n - T \bigr\|_{L^2(\nu)} \overset{n \rightarrow \infty}{\longrightarrow}  0 \).
	Assume the claim is wrong. 
	Assume the claim is wrong: 
	\[
	T_n \overset{n \rightarrow \infty}{\longrightarrow}  T_1, 
	\quad
	R_n\circ T \overset{n \rightarrow \infty}{\longrightarrow}  T_2, 
	\quad
	\|T_1 - T_2 \| > \varepsilon.
	\]
	Thus
	\[
	\bigl\|\text{id} - T_n \bigr\|_{L^2(\nu)} 
	\overset{n \rightarrow \infty}{\longrightarrow}  
	\bigl\|\text{id} - T_1 \bigr\|_{L^2(\nu)},
	\quad
	\bigl\|\text{id} - R_n\circ T \bigr\|_{L^2(\nu)}
	\overset{n \rightarrow \infty}{\longrightarrow}  
	\bigl\|\text{id} - T_2 \bigr\|_{L^2(\nu)},
	\]
	which contradicts to~\eqref{eq:convergence}
\end{proof}

The next lemma is a key ingredient in the proof of the fact
that the true kernel can be replaced by its empirical 
counterpart.
\begin{lemma} \label{lem:true}
	Consider two fixed absolutely continuous measures $\mu$ and $\nu$ in $\mathcal{W}_2(\mathbb{R}^p)$. We have a.s.
	\[
	\Bigl|
	\bigl\|T_{\mu}^{-1} - T_{\nu}^{-1} \bigr\|^2_{L^2(\mub)}
	-
	\bigl\|T^{-1}_{\mu, n} - T^{-1}_{\nu, n} \bigr\|^2_{L^2(\mun)}
	\Bigr| 
	\longrightarrow 0,~\text{as}~n\rightarrow \infty.
	\]
\end{lemma}
\begin{proof}
	Consider \(\bigl\|T^{-1}_{\mu, n} - T^{-1}_{\nu, n} \bigr\|_{L^2(\mun)} \). 
	Change of variables and triangle inequality yield
	\begin{align*}
	&\bigl\|T^{-1}_{\mu, n}  - T^{-1}_{\nu, n} \bigr\|_{L^2(\mun)}
	= 
	\bigl\|T^{-1}_{\mu, n}\circ R_n - T^{-1}_{\nu, n}\circ R_n \bigr\|_{L^2(\mub)}\\
	& \leq  \bigl\|T^{-1}_{\mu, n}\circ R_n - T^{-1}_{\mu} \bigr\|_{L^2(\mub)}
	+ \bigl\|T^{-1}_{\nu, n}\circ R_n - T_{\nu}^{-1} \bigr\|_{L^2(\mub)}
	+ \bigl\|T^{-1}_{\mu} - T_{\nu}^{-1} \bigr\|_{L^2(\mub)}.
	\end{align*}
	Therefore one obtains
	\begin{align*}
	&\bigl\|T^{-1}_{\mu, n}  - T^{-1}_{\nu, n} \bigr\|_{L^2(\mun)} 
	- \bigl\|T^{-1}_{\mu} - T_{\nu}^{-1} \bigr\|_{L^2(\mub)} \\
	&\leq
	\bigl\|T^{-1}_{\nu, n}\circ R_n - T_{\nu}^{-1} \bigr\|_{L^2(\mub)} 
	+  \bigl\|T^{-1}_{\mu, n}\circ R_n - T_{\mu}^{-1} \bigr\|_{L^2(\mub)} 
	\overset{n \rightarrow \infty}{\longrightarrow} 0
	\end{align*}
	where the last relation holds due to Lemma~\ref{lemma:comvergence_OT}.
\end{proof}

\begin{description}\item[{Proof of Proposition~\ref{p:rate}}]
\end{description}
\begin{proof}
	Actually using Lemma A.2 together with Theorem 2.2 in~\citet{kroshnin2019statistical}, we obtain that
	$$ \| R_n - {\rm Id}\|_{L^2(\mub)} =O_P\left(\frac{1}{\sqrt{n}}\right),  $$ and that the empirical transportation plan can be linearized as
	$$ T_{i,n}^{-1}=T_{i}^{-1}+D(\Sb_n-\Sb)+o(\|\Sb_n-\Sb \|_F)
	,$$ where $D$ is a linear self-adjoint bounded operator acting on the space of symmetric matrices. 
	Use the following decomposition
	\begin{align*}
	\| T_{i,n}^{-1} \circ R_n - T_i^{-1} \|_{L^2(\mub)}   \leq &  \|  T_i^{-1} \circ R_n - T_i^{-1} + (T_{i,n}^{-1}-T_i^{-1})\circ R_n \|_{L^2(\mub)} \\
	\leq & \|  T_i^{-1} \circ R_n - T_i^{-1}  \|_{L^2(\mub)} + \| (T_{i,n}^{-1}-T_i^{-1})\circ R_n \|_{L^2(\mub)} \\
	\leq & O_P\left(\frac{1}{\sqrt{n}}\right).
	\end{align*}
	This entails that $\|T_{i,n}^{-1} \ - T_{j,n}^{-1} \|_{L^2(\bar{\mu}_n)} -\|T_{i}^{-1} \ - T_j^{-1} \|_{L^2(\mub)}  $ is also of order $\frac{1}{\sqrt{n}}$ since 
	$$ \bigl\|T^{-1}_{i, n}  - T^{-1}_{j, n} \bigr\|_{L^2(\bar{\mu}_n)} - \bigl\|T^{-1}_{i} - T_{j}^{-1} \bigr\|_{L^2(\mub)}
	\leq  \bigl\|T^{-1}_{i, n}\circ R_n - T^{-1}_{i} \bigr\|_{L^2(\mub)}
	+ \bigl\|T^{-1}_{j, n}\circ R_n - T_{j}^{-1} \bigr\|_{L^2(\mub)}
	$$
	Since for all $(i,j) \in [1,N]$, $$K_{n}(i,j) = F (\|T_{i,n}^{-1}-T_{j,n}^{-1}\|^2_{L^2(\bar{\mu}_n)})$$ as soon as $F$ is continuously differentiable with bounded derivative, then we get that for a finite constant
	$$ \sum_{i,j=1}^N | K_{n}(i,j) - K(i,j) |^2 \leq N^2 \sup_{i,j} | K_{n}(i,j) - K(i,j) |^2 \leq C \frac{N^2}{n},$$ which concludes the proof.
\end{proof}

\begin{description}\item[{Proof of Proposition~\ref{prop:stationarity}}]
\end{description}
\begin{proof}
	Note, that for any orthogonal matrix \(U\) the following set of inequalities hold:
	\begin{align*}
	W^2_2\left(\mathcal{N}(0, S),  \mathcal{N}(0, Q)\right)  & := \text{tr}(S) + \text{tr}(Q) - 2 \text{tr} \left(Q^{1/2}S Q^{1/2} \right)^{1/2}\\
	& =  W^2_2\left(\mathcal{N}(0, USU^{T}),  \mathcal{N}(0, UQU^{T})\right)\\
	& =  W^2_2\left(\phi_{U}(\mathcal{N}(0, S)),  \phi_{U}(\mathcal{N}(0, Q))\right).
	\end{align*}
	Thus, map \(\phi_U\) preserves \(2\)-Wasserstein distance.
	Equality~\eqref{eq:stationarity} follows from ~\eqref{eq:OT_map_norm} by substituting \(S_i\), \(S_j\), and \(\Sb\) by
	\(US_iU^{T}\), \(US_jU^{T}\), and \(U\Sb U^{T}\) respectively.
\end{proof}

\end{document}